\documentclass{llncs}

\usepackage{graphicx}
\graphicspath{{./Figures/}}

\usepackage{amsmath}
\usepackage{amssymb}
\usepackage{dsfont}
\usepackage{array}
\usepackage{multirow}
\usepackage{subfig}

\newtheorem{fact}{Fact}

\newcommand{\SAT}{\textsc{3-CNF-SAT}}
\newcommand{\QBF}{\textsc{3-CNF-QBF}}
\newcommand{\scrabble}{\textsc{Scrabble}}
\newcommand{\scrabblesol}{\textsc{Scrabble-Solitaire}}

\newcommand{\test}[3]{c_{#2}^{#3 #1}}  

\newcommand{\bag}{\sigma}
\newcommand{\position}{\pi}
\newcommand{\play}{\Pi}
\newcommand{\turn}{p}
\newcommand{\board}{\mathcal{B}} 
\newcommand{\rack}[1]{r^{#1}}
\newcommand{\score}[1]{s^{#1}}
\newcommand{\s}[1]{\emph{\textbf{#1}}} %Use this command for symbols

\begin{document}

\title{Scrabble is PSPACE-Complete}

\author{Michael Lampis\inst{1} \and Valia Mitsou\inst{2} \and Karolina So{\l}tys\inst{3}}
\institute{KTH Royal Institute of Technology, \email{mlampis@kth.se} \and Graduate Center, City University of New York, \email{vmitsou@gc.cuny.edu} \and Max Planck Institute f\"{u}r Informatik, \email{ksoltys@mpi-inf.mpg.de}}

\maketitle

\begin{abstract} In this paper we study the computational complexity of the
game of Scrabble. We prove the PSPACE-completeness of a derandomized model of
the game, answering an open question of Erik Demaine and Robert Hearn.


\end{abstract}

\section{Introduction}\label{introduction}

In this paper we examine the computational complexity of optimal play in the
game of Scrabble, a board game played by two to four players. In this game each
player takes turns drawing lettered tiles randomly out of a bag and then
attempting to place those tiles on a common $15\times 15$ board, forming words.
Points are awarded depending on the length of the formed words, the value of
the letters used and various bonuses found on the board, with the winner being
the player who has gathered the highest number of points at the end of the
game.\footnote{For a fuller description of the board game of Scrabble see e.g.\
\url{http://en.wikipedia.org/wiki/Scrabble/}} 

Having been invented in the US around the middle of the 20th century, Scrabble
is now one of the most popular and well-known board games in the world. Besides
the original english language version, Scrabble has been translated to dozens
of other languages, while more than one hundred million Scrabble sets have been
sold worldwide.

Since Scrabble is such a successful game, it becomes a natural question to
determine the computational complexity of finding an optimal play. Similar
questions have already been answered for several other popular board games,
such as Othello, Chess and Go, typically classifying their complexity as either
PSPACE or EXPTIME-complete. This is, however, complicated by the fact that,
unlike those games, chance plays a non-negligible part in a match of Scrabble,
as players don't know in advance the order in which tiles will be drawn. Still,
much insight could be gained by investigating the complexity of a
perfect-information version of Scrabble, where the order in which tiles will be
drawn is known beforehand. In fact, this was listed as an open problem by
Demaine and Hearn \cite{demaine2005playing}. This is exactly the question we
tackle in this paper by showing that this derandomized version of Scrabble is
PSPACE-complete.

This result on its own is probably not surprising, since most interesting board
games are at least PSPACE-hard, and Scrabble is trivially in PSPACE from the
fact that tiles cannot be removed from the board once they are placed. In
addition to settling the complexity question though, we go about trying to
understand what exactly makes the problem hard.

Informally, at any given round a Scrabble player is confronted with two tasks:
deciding which word to form and deciding where to place it on the board. Though
the tasks are not independent, since the formed word must be using some tiles
already on the board, they are conceptually different and the hardness of the
game could stem from either one. Put another way, it could be the case that
deciding which word is best to play is easy if there is only one possible
position where a word can be placed, or that deciding where to place the next
word is easy if only one word can be made with the available tiles.

In fact, we will present two different hardness proofs arguing that both of
these tasks are hard. In one reduction the players will be given appropriate
tiles so that they will only have one possible word to play in each round, with
a choice of two places to place it. In the other, players will be forced to
play in a specific place on the board, but will be able to choose between two
different words. In both cases, the problem of deciding optimal play will still
turn out to be PSPACE-complete. Along the way, we can show that even a
single-player version of the game, where one player tries to place all tiles,
is NP-complete in both cases. Thus, we establish that during the course of a
game, Scrabble players need to perform not one, but two computationally hard
tasks, which is probably the reason why Scrabble is so much fun to play.

\section{Our model of Scrabble - Definitions}\label{model}

Informally, the question we are trying to answer is, given a Scrabble position
how hard is it to determine the best playing strategy? As mentioned, we will
tackle this problem in a perfect information setting, where the contents of the
bag and the order in which they are drawn are known in advance to both players
(and therefore both players know each other's letters).

Moreover, since Scrabble is a finite game, in order to study its computational
complexity we need to consider some unbounded generalization. The most natural
way to go forward is to consider the game played on an $n\times n$ board. In
addition, we assume that the bag initially contains a number of tiles that
depends on $n$, since the restriction of the game where the bag contains a
fixed number of tiles will yield at most a polynomial number of possible
configurations, putting the problem trivially in P.

Beyond the size of the board and the number of letters in the bag, we need to
define an alphabet, a set of acceptable words and a rack size which will
determine how many letters each player has on hand. All of these can be allowed
to depend on the input, but since we are interested in proving hardness results
we are happier when we can establish them even if those parameters are fixed
constants. In fact, in Theorem \ref{thm:constant-size} we prove that Scrabble
is PSPACE-hard even with these restrictions, at the cost of making the
reduction a little technical.

We will deal with a plain version of the game, where all letters have the same
value and there are no premium positions on the board (clearly, the more
general case with multiple values and possible premiums is harder). Also, for
the most part we will assume that players are not allowed to exchange tiles or
pass. Nevertheless, we will give arguments after Theorem
\ref{thm:constant-size}  explaining why allowing players to pass does not
affect our results.

Let us now give a more formal definition of the problem:

\begin{definition}

We define a \emph{Scrabble game} $\mathcal{S}$ to be an
ordered quadruple $(\Sigma, \Delta, k, \\ \position_0)$ where: $\Sigma $ is a
finite \emph{alphabet}, $\Delta \subset \Sigma^*$ is a finite
\emph{dictionary}, $k \in \mathds{N}$ is the size of the rack and $\position_0$
is the initial position of the game, defined as below.

%\begin{itemize}
%
%   \item $\Sigma $ is a finite \emph{alphabet}; 
%   \item $\Delta \subset \Sigma^*$ is a finite \emph{dictionary}; 
%   \item $k \in \mathbb{N}$ is the size of the rack; 
%   \item $\position_0$ is the initial position of the game, defined as below.
%
%\end{itemize} 

\end{definition}

\begin{definition} 

A \emph{position} $\position$ in a scrabble game is an ordered septuple 
$(\board, \bag, \turn, \rack{1},\\ \rack{2}, \score{1}, \score{2})$, where 
$\board \in \mathbf{M}_{n \times n}(\Sigma)$ is the \emph{board}, 
$\bag \in \Sigma^*$ is a sequence of lettered tiles called the 
\emph{bag}, 
$\turn \in \{1,2\}$ is the number of the active player,
$\rack{i}, \text{where } i \in \{1,2\}$, are multisets with symbols 
from $\Sigma$ denoting the contents of the rack of the first and the 
second player respectively and $\score{i} \in \mathds{N}, \text{where } i \in
\{1,2\}$, are the scores of the first and the second player respectively. 

%\begin{itemize}
%   
%   \item $\board \in \mathbf{M}_{n \times n}(\Sigma)$ is the \emph{board}; 
%   \item $\bag \in \Sigma^*$ is a sequence of lettered tiles called the 
%         \emph{bag}; 
%   \item $\turn \in \{1,2\}$ is the number of the active player;
%   \item $\rack{i}, \text{where } i \in \{1,2\}$, are multisets with symbols 
%         from $\Sigma$ denoting the contents of the rack of the first and the 
%         second player respectively;
%
%	 %?????????: Should the rack be defined as a string or a set?
%     
%   \item $\score{i} \in \mathbb{N}, \text{where } i \in \{1,2\}$, are the scores 
%         of the first and the second player respectively. 
%
%\end{itemize}

\end{definition}

\begin{definition}

A \emph{play} $\play = \position_1 \dots \position_l$ is a sequence of
positions such that, for all $i$, $\position_{i+1}$ is attainable from
$\position_i$ by the active player by forming a \emph{proper play} on the
board.  

\end{definition}

A proper play uses any number of the player's tiles from the rack to form a
single continuous word (\emph{main word}) on the board, reading either
left-to-right or top-to-bottom. The main word must either use the letters of
one or more previously played words, or else have at least one of its tiles
horizontally or vertically adjacent to an already played word. If words other
than the main word are newly formed by the play, they are scored as well, and
are subject to the same criteria for acceptability. All the words thus formed
must belong to the dictionary. After forming a proper play, the sum of the
lengths of all words formed is added to the active player's points, letters
used are removed from the player's rack and the rack is refilled up to $k$
letters (or less, if $|\sigma_i|<k$) with the appropriate number of letters
forming the prefix of $\sigma_i$.

%\footnote{In the original verion of Scrabble the two players are allowed to
%pass, or to exchange tiles if they cannot form a word with their current
%selection of letters. We omit this detail for the sake of the simplicity of
%the model. Slight modifications to the proofs can show the hardness of the
%original model}

\begin{definition} 

A play $\play = \position_1 \dots \position_l$ is \emph{finished} if player \
$l+1 \bmod 2$ is unable to form a proper play, or if $\bag_l = \varepsilon$
(i.e.  the bag is empty). The \emph{winner} of a finished play is the player
with the greater number of points (draws are possible).

\end{definition}

%\begin{definition}
%
%A \emph{Scrabble solitaire} game is defined analogously to the normal game, but
%with only a single player. The player \emph{solves} the solitaire if she
%manages to get rid of all the letters from the bag. We define \emph{$(\Sigma,
%\Delta, k, \position)$-Scrabble solitaire}, with $\Sigma$, $\Delta$ and $k$ as 
%above and $\position = \{\board, \bag, \rack{1}\}$, with $\board$, $\bag$ and 
%$\rack{1}$ defined as above.  
%
%\end{definition} 
%
%\begin{definition} 
%
%We define \scrabble\ to be the problem of determining the winner of a given
%Scrabble game and \scrabblesol\ to be the problem of determining if a given
%Scrabble solitaire game is solvable.
%
%\end{definition}

We will establish PSPACE-hardness via two reductions from \QBF, the problem of
deciding whether a quantified boolean formula is true. This is a well-known
PSPACE-complete problem often used to establish hardness for games
\cite{papadimitrioucomputational}. We are also interested in the variation of
the game where there is only one player who tries to place all the tiles on the
board, which we call \scrabblesol. Essentially the same constructions we
present can also establish NP-hardness for \scrabblesol\ if one begins the
reduction from \SAT.

\section{Hardness due to placement of the words}\label{placement}
%In the proof we will see that it is sufficient to take $\Sigma$ such that
%$|\Sigma| \leq 10 $ and $k \leq 7 $, which means that the reduction can be
%implemented with standard English-language scrabble set by using only 1-point
%letters (there are 10 such letters in the English language set). 

In this section we prove that \scrabble\ is PSPACE-complete due to ability of
players to place their formed word in more than one places.\footnote{In this
section we prove hardness of a version of \scrabble\ with an unbounded size
alphabet. In section \ref{formation} we prove the hardness of the natural
variant of derandomized \scrabble, where the alphabet, word, rack and
dictionary sizes are constants.} 

%\subsection{Outline of the proof}\label{outline}

We will first prove that the one-player version \scrabblesol\ is NP-complete.
PSPACE-completeness of \scrabble\ follows with slight modifications.

%\subsection{NP-completeness of Scrabble solitaire}\label{np} 

%In this section we will prove the following lemma.  

\begin{lemma}\label{lemma:np} 

\scrabblesol\ is NP-complete.

\end{lemma} 

%Proving that the problem is in NP is straightforward

% -- to check whether a game is solvable we simply guess the sequence of moves
%leading to solving the game.  The length of this sequence will be polynomial
%(in the terms of the size of the description of the game), and each move can be
%described by some polynomial information.

Proving that the problem is in NP is straightforward. To estabilish the
NP-hardness of \scrabblesol, we will construct a reduction to this problem from
\SAT. Given 3-CNF propositional formula $\phi$ with $n$ variables $x_1, x_2,
\ldots, x_n$ and $m$ clauses, we construct in polynomial time a
polynomial-sized Scrabble-Solitaire game $\mathcal{S}$, such that $\phi$ is
satisfiable iff $\mathcal{S}$ is solvable.

The general idea of the proof is as follows. We will create gadgets associated
to variables, where the player will assign values to these variables.  We will
ensure that the state of the game after the value-assigning phase completes,
will correspond to a consistent valuation. Then the player will proceed to the
testing phase, when for each clause she will have to choose one literal from
this clause, which should be true according to the gadget of the respective
variable.  If she cannot find such a literal, she will be unable to complete a
move.  Thus we will obtain an immediate correspondence between the
satisfiability of the formula and the outcome of the game.

%\begin{figure}  
%\centering
%\includegraphics[scale=0.3]{variables_placement} 
%\caption{A gadget corresponding to the variable $x_i$.} 
%\label{gadgetfig}
%\end{figure}

%Initially, the only letters found on the board are those on the perimeter,
%that is, all the characters shown except the \s{\$}'s and the inner
%$\mathbf{x_i}$.

The gadget for variable $x_i$ is shown in Figure \ref{var_place} in the
appendix. The construction of the dictionary and the sequence in the bag will
ensure that at some point during the value-assigning, the only way for the
player to move on is to form a word like in Figure \ref{assgn:false} or to form
a horizontally symmetrical arrangement (Fig. \ref{assgn:true}).  

%\begin{figure}
%\centering
%\includegraphics[scale=0.3]{assignment_placement}
%\caption{Variable $x_i$ with an assigned value.} \label{fig:assignments}
%\label{assign:placement}
%\end{figure}

\begin{figure} 
%\centering 
%\subfloat[Gadget for variable $x_i$.]
%{\label{gadgetfig}\includegraphics[scale=0.15]{variables_placement} }
\subfloat[$x_i$ set to false.]
{\label{assgn:false}\includegraphics[width=0.35\textwidth]{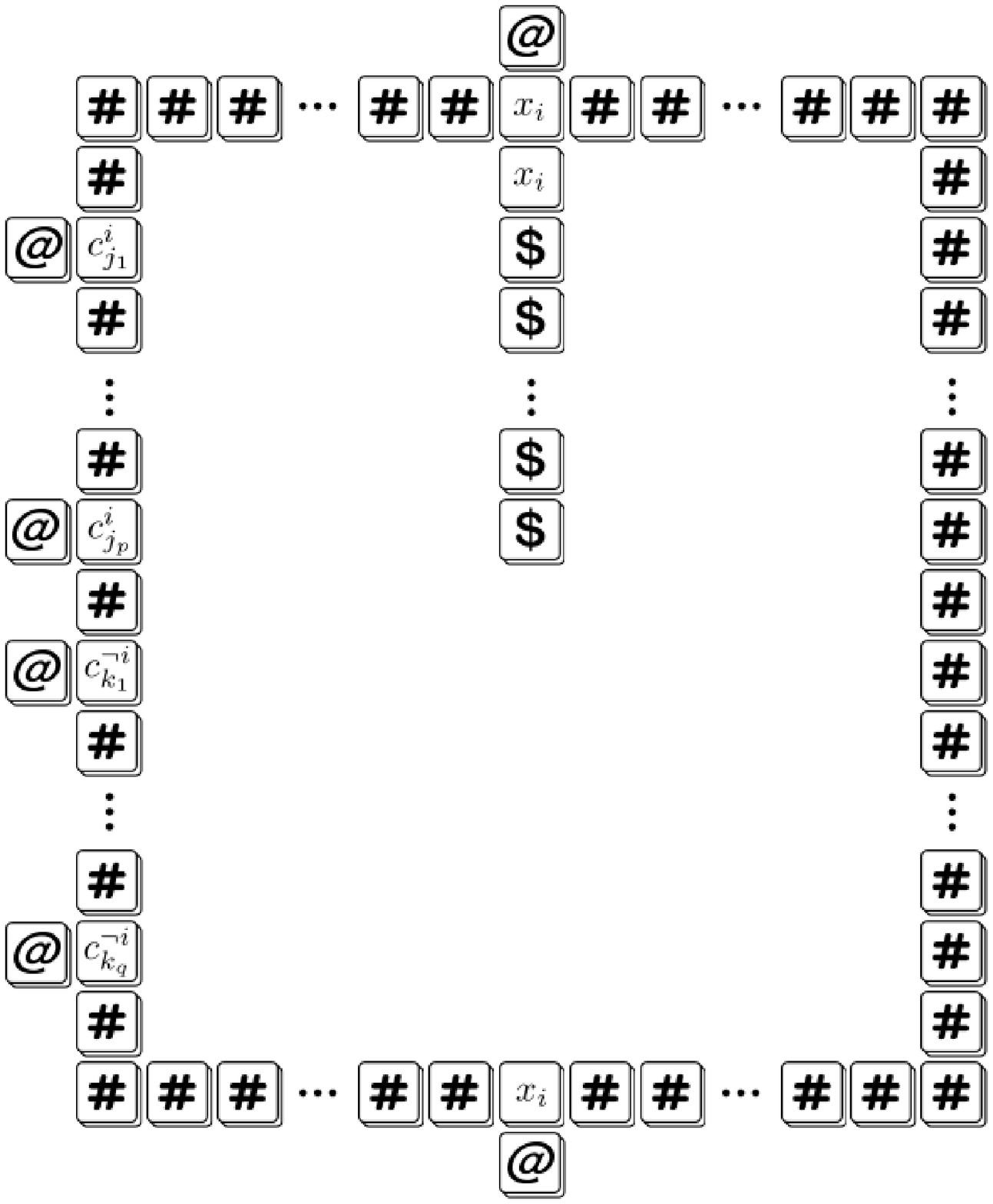}}
\subfloat[$x_i$ set to true.]
{\label{assgn:true}\includegraphics[width=0.35\textwidth]{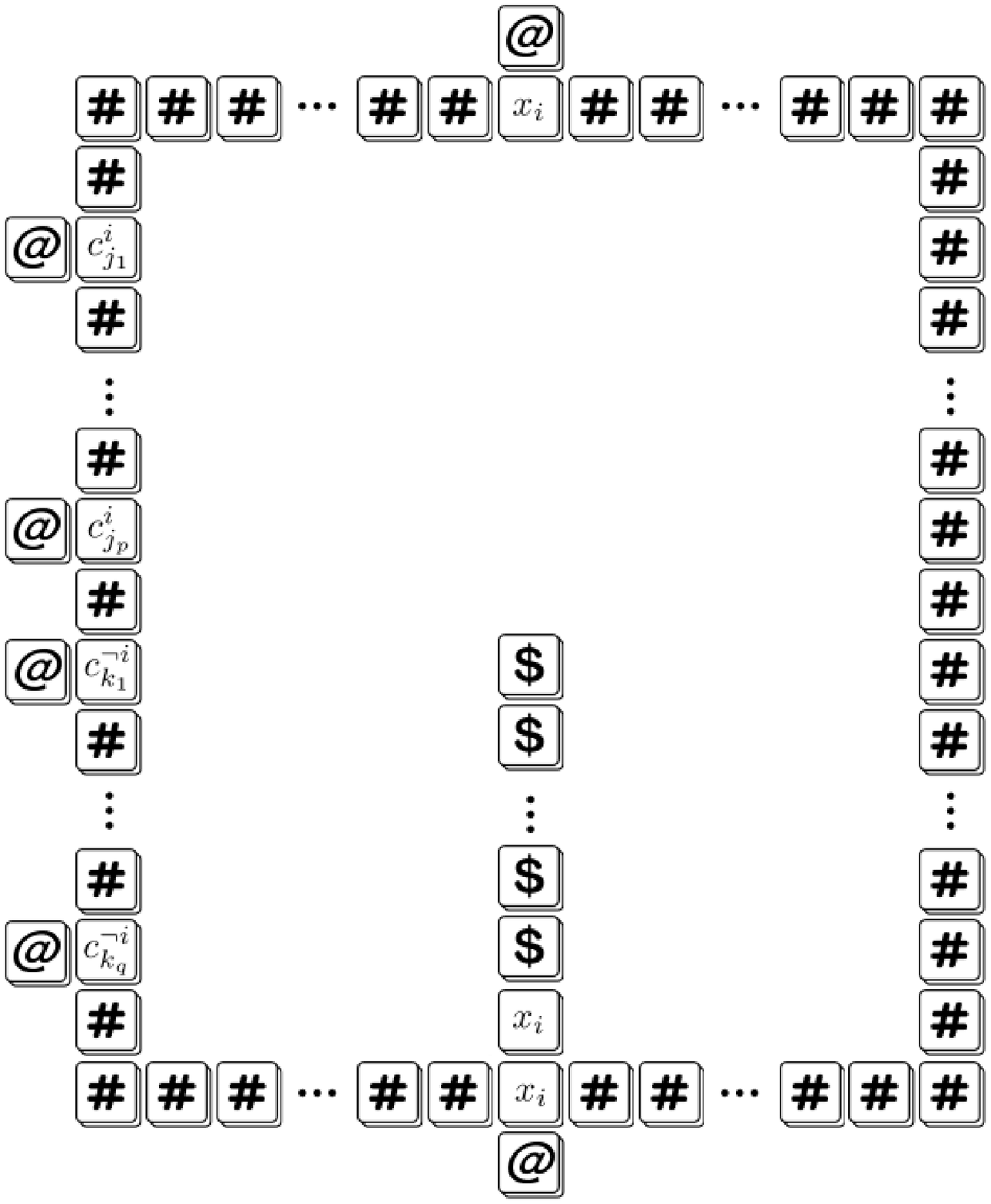}}
%\subfloat[A word corresponding to the clause $(\neg x_i \vee \neg x_e \vee x_g)$ 
%can be played in $x_i$'s gadget, because the value-assigning word for $x_i$ has 
%been played to indicate $\neg x_i$.]
\subfloat[A word played for a clause that $\neg x_i$ satisfies]
{\label{testfig} \includegraphics[width=0.35\textwidth]{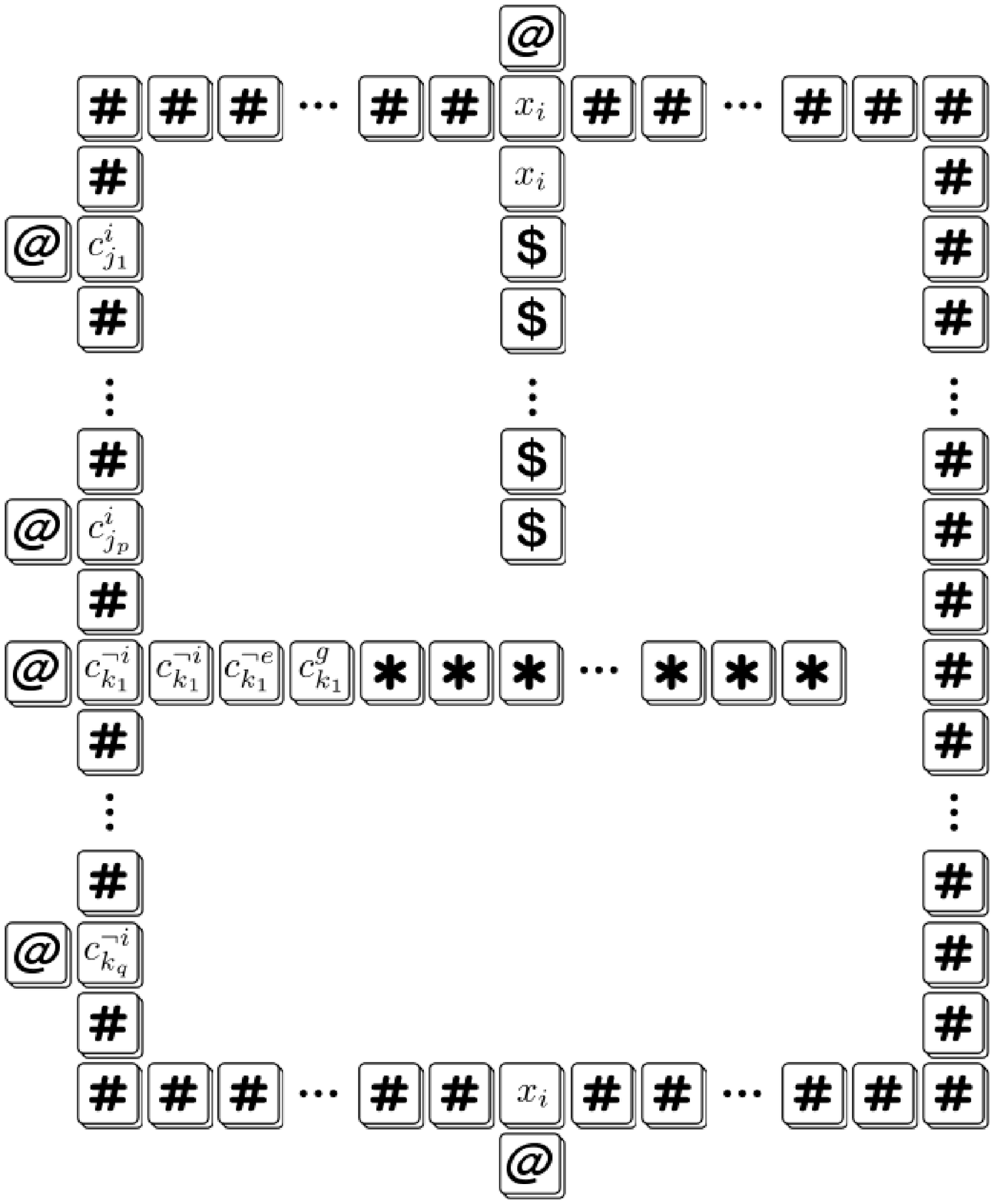}}
\caption{Variable $x_i$ with an assigned value.} \label{fig:assignments}
\end{figure}

During the test phase, for each clause $c_i = (l_1 \vee l_2 \vee l_3)$ in every
play there will be a position, when the player will be obliged to choose one of
the literals from the clause, in whose gadget she will try to play a word. She
will be able to form a word there iff the value of the corresponding variable,
which has been set in the earlier phase, agrees with the literal.

%\begin{figure}[h]
%\centering
%\includegraphics[scale=0.5]{satisfaction_placement}
%\caption{A word corresponding to the clause $(\neg x_i \vee \neg x_e \vee x_g)$ 
%can be played in $x_i$'s gadget, because the value-assigning word for $x_i$ has 
%been played to indicate $\neg x_i$.}
%\label{testfig}
%\end{figure}

Let us describe the game more formally. The alphabet $\Sigma$ of $\mathcal{S}$
will contain: 
\begin{itemize} 
   
   \item a symbol \s{$x_i$} for every variable $x_i$; 
   \item a symbol \s{$\test{i}{j}{}$} (\s{$\test{i}{j}{\neg}$}), for every 
	 positive (negative) appearence of variable $i$ in the clause $j$; 
   \item auxilliary symbols: \s{\$}, \s{\#}, \s{$\ast$} and \s{@}.  

\end{itemize} 

Let $r$ be such that no literal appears in more than $r$ clauses. The rack size 
will be $k=2r$. 

%\footnote{\SAT\ remains NP-hard even in the restricted version where every 
%literal appears at most 2 times, so $k$ can be set to equal 4. 

The dictionary $\Delta$ will contain the following words:

\begin{itemize}

\item the words \s{@}\s{$x_i$}\s{$x_i$}\s{\$}$^{2r-1}$ and 
      \s{\$}$^{2r-1}\s{$x_i$}\s{$x_i$}\s{@}$
      for every variable $x_i$, 
\item the word \s{@}\s{$\test{a}{j}{(\neg)}$}\s{$\test{a}{j}{(\neg)}$}\s{$\test{b}{j}{(\neg)}$}\s{$\test{c}{j}{(\neg)}$}\s{$\ast$}$^{2r-3}$,
      for every permutation $(a, b, c)$ of the indices of the literals appearing 
      in the clause $c_j$.  
\item We also add all the dummy words appearing initially on the board.

%We also add the following dummy words \s{\#}$^{4r+3}$,
%			\s{@}\s{$x_i$}, 
%			\s{$x_i$}\s{@}, 
 %     			\s{\#}$^r$\s{$x_i$}\s{\#}$^r-1$,
%			\s{\#}$\prod_{u=1}^p$(\s{\#}\s{$\test{i}{j_u}{}$})\s{\#}$\prod_{v=1}^p$(\s{$\test{i}{k_v}{\neg}$}\s{\#})\s{\#}
%		and \s{@}\s{$\test{i}{j}{(\neg)}$} 
%for every variable $x_i$, and for every
%appearence of variable $i$ in the clauses.  These words appear on the board in
%the beginning of the game.

\end{itemize} 

The sequence in the bag $\bag$ will be a concatenation of the following: 

$$ \bag = \prod_{i=1}^n \left(x_i \s{\$}^{2r-1} \right)
\prod_{j=1}^m\left(\test{a}{j}{(\neg)}\test{b}{j}{(\neg)}\test{c}{j}{(\neg)}\s{$\ast$}^{2r-3}\right)$$

The time period, when at least one of the letters $x_i$ are still on the rack 
will be called the \emph{value-assigning phase}. The following time period will 
be called the \emph{satisfaction phase}. 

We can now prove the following facts (omitted proofs appear in the appendix).

\begin{fact}\label{eatsall}

The player has always to empty her rack in order to perform a proper
play.

\end{fact}

\begin{fact} \label{fact2}

During the value-assigning phase, at each turn the player performs an action
that is in our setting equivalent to a correct valuation of a variable, as
shown in Figure	\ref{fig:assignments}.

\end{fact}

\begin{fact} \label{fact3} 

During the test phase, at each turn the player's actions are equivalent to
checking whether a clause, that had not been checked before, is satisfied by a
literal of the player's choice, as shown in Figure \ref{testfig}.  

\end{fact} 

\begin{proof}

Basing on the previous two facts we know that during each round in the
satisfaction phase, the contents of the player's rack are
$\{\test{a}{j}{(\neg)}, \test{b}{j}{(\neg)}, \test{c}{j}{(\neg)}$,
$\ast^{2r-3}$\} for $a$, $b$ and $c$ being the indices of the literals
appearing in clause $j$. One can easily see that the player can form a legal
word from these letters only by extending one of the 3 words
\s{@}$\test{v}{j}{(\neg)}$, $v \in \{a, b, c\}$, by arranging her symbols in a
permutation $(a', b', c')$ where $v = a'$.

The player can choose any of such permutations, which means she can choose the
literal, in whose gadget she will play the word. A simple analysis shows that
the player can play this word in that position iff the valuation of the
variable agrees with the chosen literal (i. e. if the chosen literal reads
$\neg x_j$, then $x_j$ must have been set to false etc.).

\end{proof}

The above facts imply that the game correctly simulates assigning some
valuation to a 3-CNF formula and checking whether it is satisfied. It is easy
to check that the size instance of the Scrabble solitaire game obtained by the
reduction is polynomial in terms of the size of the input formula and that the
instance can be computed in polynomial time. We have thus shown that
\scrabblesol\ is NP-complete.  

%\subsection{PSPACE-completeness of Scrabble}\label{pspace}

To prove the PSPACE-completeness of \scrabble\ it suffices to notice that the
above reduction from \SAT\ to \scrabblesol\ easily translates to the analogous
reduction from \QBF (a detailed proof of the following theorem can be found in 
the appendix).

\begin{theorem} \label{thm:pspace1}

\scrabble\ is PSPACE-Complete.

\end{theorem}

\section{Hardness due to formation of the words}\label{formation}
In this section we prove the hardness of Scrabble due to the ability of the
players to form more than one words using the same letters. Furthermore, we
will optimize this reduction so that it works even for constant-size $\Sigma,
\Delta$ and $k$.

\begin{theorem} \label{thm:constant-size} 

\scrabble\ is PSPACE complete even when restricted to instances with
constant-size alphabet, dictionary and rack.

\end{theorem} 

\begin{proof}

We will proceed in steps. In section \ref{sec:sketch} we simply sketch the
high-level idea, which consists of a board construction that divides play into
two phases, the assignment and the satisfaction phase. Then, in sections
\ref{sec:initial}, \ref{assignment}, \ref{satisfaction} we
present in full a slightly simplified version of our construction which uses a
constant-size $\Sigma$ and $\Delta$ but unbounded $k$. Finally, in section
\ref{sec:constant} we give the necessary modifications to remove words of
unbounded length from the dictionary and obtain a reduction where $k$ is also
constant.

\subsection{Construction Sketch} \label{sec:sketch}

Our reduction is from \QBF. Suppose that we have a 3-CNF-QBF 
formula $\exists x_1 \forall x_2 \exists x_3 \ldots\phi$ with $n$ variables
$x_1, x_2, \ldots, x_n$, where $\phi$ has $m$ clauses $c_1, c_2, \ldots, c_m$.
We create an instance of $(\Sigma,\Delta,k,\position)$-SCRABBLE, as follows.

The board will be separated in $n$ roughly horizontal segments which correspond
to variables and $m$ vertical segments which correspond to clauses (see figure
\ref{highlevel}). 

\begin{figure}
\centering
\includegraphics[scale=0.3]{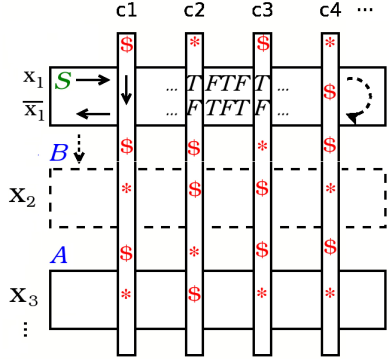}
\caption{A high level of the game}
\label{highlevel}
\end{figure}

%\begin{figure}[h]
%\centering
%
%\subfloat[A high level view of the game.]
%{ \label{highlevel} \includegraphics[width=0.5\textwidth]{epsfigures/highlevel} }
%\quad
%\subfloat[Sketch of construction. Red characters indicate clause columns.]
%{\label{board} \includegraphics[width=0.5\textwidth]{epsfigures/board-restricted} }
%\caption{An abstract view of the board.}  
%
%   \end{figure}

Play will be divided into two phases: the assignment phase and the satisfaction
phase. In the first phase the two players will play within the horizontal
segments placing words that encode the truth values of the variables of the
formula (hence, mostly the letters T and F are used in this phase). With
appropriately placed walls we keep the players on track in this phase making
sure that each player, during her turn, has only one available position to
place a word (but possibly two availabe words to place if it is her turn to
decide on a variable's truth value).

For the second phase, the players place words in the vertical segments. Here,
we have encoded the structure of the formula by placing a different character
on the intersection of two segments depending on whether the corresponding
literal appears in that clause. The first player is always forced to play a
word in these intersections and she is only able to do so till the end if and
only if the chosen truth assignment satisfies all clauses.

\subsection{The initial position} \label{sec:initial}

Let us now describe the game more formally.  We create a $(\Sigma, \Delta, k,
\position)$ game of Scrabble, where the alphabet $\Sigma =$ \{\s{\#}, \s{*},
\s{\$}, \s{A}, \s{B}, \s{S}, \s{T}, \s{F}, \s{0}, \s{1}, \s{2}, \s{@}\}, the
rack size $k$ is an odd number depending on $m$ (particularly $k=10m+5$),
$\Delta$ is shown in table \ref{dictionary} and the initial position
$\position$ is described below.

\begin{table}
\centering
\footnotesize
\begin{tabular}{|l|l|} %p{7cm}|}

\hline \multicolumn{2}{|c|}{Dictionary} \\ 

\hline \hline 

Words & Definition \\ 

\hline

\s{S}(\s{T}\s{F})$^\frac{k-1}{2}$\s{S}, 
\s{F}(\s{T}\s{F})$^\frac{k-1}{2}$\s{S}, \s{S}(\s{F}\s{T})$^\frac{k-1}{2}$\s{F,} & 
\multirow{2}{*}{\emph{The literal played has value True.}} \\
\s{F}(\s{T}\s{F})$^\frac{k-3}{2}$\s{S}\s{T}\s{F}\s{T}\s{F},
\s{F}(\s{T}\s{F})$^\frac{k-3}{2}$\s{S}\s{F}\s{T}\s{F}\s{T} & \\

\hline   
   
\s{S}(\s{F}\s{T})$^\frac{k-1}{2}$\s{S}, 
\s{T}(\s{F}\s{T})$^\frac{k-1}{2}$\s{S}, \s{S}(\s{T}\s{F})$^\frac{k-1}{2}$\s{T}, & 
\multirow{3}{*}{\emph{The literal played has value False.}} \\ 
\s{T}(\s{F}\s{T})$^\frac{k-3}{2}$\s{S}\s{T}\s{F}\s{T}\s{F},
\s{T}(\s{F}\s{T})$^\frac{k-3}{2}$\s{S}\s{F}\s{T}\s{F}\s{T} & \\
      
\hline

\s{\#}\s{A}\s{T}, \s{\#}\s{A}\s{F} & \emph{First player's turn to assign truth value}; \\ 
%\s{\#}\s{A}\s{F} & 2. \emph{Winning word for player 1} \\ 

\hline
   
\s{\#}\s{B}\s{S} &  \emph{Second player's turn to 
assign truth value}; \\ 
% & 2. \emph{Player is forced to place his \s{S} symbol. in the specified 
%position}; \\ 

\hline

\s{\$}\s{\$}, \s{*}\s{*}, \s{\#}\s{A}, \s{\#}\s{B}, \s{\#}$^c$, for $c\le 2k$ 
& \multirow{2}{*}{\emph{Wall word}} \\   
\s{\#}$^5$Q\s{\#}$^9Q$\s{\#}$^9Q$\s{\#}$^5$, for $Q \in \{\s{\$},\s{*}\}$ & \\
      
\hline  

%\s{0}\s{*}\s{*}, \s{1}\s{*}\s{*}, \s{2}\s{*}\s{*} & \multirow{2}{*}{\emph{Word 
%formed during satisfaction phase}.} \\
%\s{0}\s{\$}\s{\$}, \s{1}\s{\$}\s{\$}, \s{2}\s{\$}\s{\$} & \\
         
\s{0}\s{*}\s{*}, \s{1}\s{*}\s{*}, \s{2}\s{*}\s{*}, \s{0}\s{\$}\s{\$}, 
\s{1}\s{\$}\s{\$}, \s{2}\s{\$}\s{\$} &
\emph{Word formed during satisfaction phase}. \\

\hline

\s{0}\s{*}\s{*}\s{1}\s{T}\s{2}\s{0}, \s{0}\s{\$}\s{\$}\s{1}\s{T}\s{2}\s{0}, 
\s{0}\s{\$}\s{\$}\s{1}\s{F}\s{2}\s{0} 
& \emph{No unsatisfied literals in the clause so far}. \\
%\s{0}\s{\$}\s{\$}\s{1}\s{T}\s{2}\s{0} & \\
       
\hline

\s{1}\s{*}\s{*}\s{2}\s{T}\s{0}\s{1}, \s{1}\s{\$}\s{\$}\s{2}\s{T}\s{0}\s{1}, 
& \multirow{2}{*}{\emph{One unsatisfied literal in the clause so far}.} \\
%\s{1}\s{\$}\s{\$}\s{2}\s{T}\s{0}\s{1} & \\
\s{1}\s{\$}\s{\$}\s{2}\s{F}\s{0}\s{1}, \s{0}\s{*}\s{*}\s{2}\s{F}\s{0}\s{1} & \\
%\s{0}\s{*}\s{*}\s{2}\s{F}\s{0}\s{1} & \\
       
\hline

\s{2}\s{*}\s{*}\s{0}\s{T}\s{1}\s{2}, \s{2}\s{\$}\s{\$}\s{0}\s{T}\s{1}\s{2}, & \multirow{2}{*}{\emph{Two unsatisfied 
literals in the clause so far}.} \\
%\s{2}\s{\$}\s{\$}\s{0}\s{T}\s{1}\s{2} & \\
\s{2}\s{\$}\s{\$}\s{0}\s{F}\s{1}\s{2}, \s{1}\s{*}\s{*}\s{0}\s{F}\s{1}\s{2} & \\
%\s{1}\s{*}\s{*}\s{0}\s{F}\s{1}\s{2} & \\
        
\hline
 
\s{0}\s{1}\s{2}\s{0},  \s{1}\s{2}\s{0}\s{1}, \s{2}\s{0}\s{1}\s{2} & \multirow{1}{*}{\emph{Symbols' \s{0}, \s{1}, \s{2} order 
preserving words}.} \\
%\s{1}\s{2}\s{0}\s{1} & \\  
%\s{2}\s{0}\s{1}\s{2} & \\
 
\hline

\end{tabular}      

\smallskip
\caption{The Dictionary $\Delta$. All valid words appear as regular expressions,
together with their definitions. Synonyms are grouped together.}
\label{dictionary}

\end{table}

For the following descriptions refer to figure \ref{highlevel} (or for a more 
detailed but still abstract preview to figure \ref{board} in the appendix).

The initial board $\board$ consists mainly of words containing the dummy symbol
\s{\#}. We use these words to build walls inside the board that will restrict
the players' available choices. 

There is also a symbol \s{S} initially placed on the board. This indicates the
starting point, where the first player is going to put her first word.

On the left side of the board, attached on the wall, there are several
appearences of the symbols \s{A} and \s{B} (shown in blue). These symbols
indicate whether it is the first or the second player's turn to choose truth
assignment (player 1 assigns values to the variables $x_{2i+1}$ whereas player
2 to the variables $x_{2i}$ for every $i = \lfloor \frac{n}{2} \rfloor$). 

%There are some more appearences of the symbol \s{B} on the board, attached on
%the left side of the horizontal walls corresponding to negated literals (not
%visible in figure \ref{board}, check a more detailed figure, for example
%figure \ref{variables}). They are placed here because we want to force player
%1 to play a specific word from the dictionary among the two possible choices
%that she would have if symbol \s{B} wasn't there.

Last, we need to construct the clauses. For every clause there is a
corresponding column as shown in the figure. We place the symbols \s{\$} and
\s{*} in the intersections with literals (horizontal lines) in order to
indicate which literals appear in the particular clause (if a literal appears
in the clause we put a \s{*} whereas if it doesn't we put a \s{\$}). In the
figure \ref{highlevel}, $c_4 = (x_1 \vee \neg x_2 \vee \neg x_3)$ 

%Figure \ref{board} shows a very abstract version of the initial board
%construction for the propositional formula $\phi = (x_1 \vee \neg{x_2} \vee
%\neg{x_3}) \wedge (\neg{x_1} \vee x_2 \vee x_4) \wedge (\neg{x_2} \vee x_3
%\vee \neg{x_4})$. More details are visible in the following figures
%\ref{variables} and \ref{clauses}.  

In the initial position $\position$ of the game we also have:

\begin{itemize} 

   \item $\rack{1} = \rack{2} = r = \{\s{T}, \s{F}\}^\frac{k-1}{2} \cup \{\s{S}$\};
   \item $\bag = r^a (\s{0}\s{1}\s{2})^s\s{@}^{2k-6}\s{A}$, where $a\text{ }(=4n-2)$ is 
         the number of turns played during the assignment phase and 
         $s\text{ }(=\frac{40}{3}m^2n)$ the number of turns played during the satisfaction 
         phase (see sections \ref{assignment} and \ref{satisfaction});
   \item Player 2 has a lead of 1 point and it is first player's turn.

\end{itemize} 

\subsection{Assignment Phase}\label{assignment}

\begin{figure}
\subfloat[The assignment phase]
{\label{variables}
\includegraphics[width=0.5\textwidth]{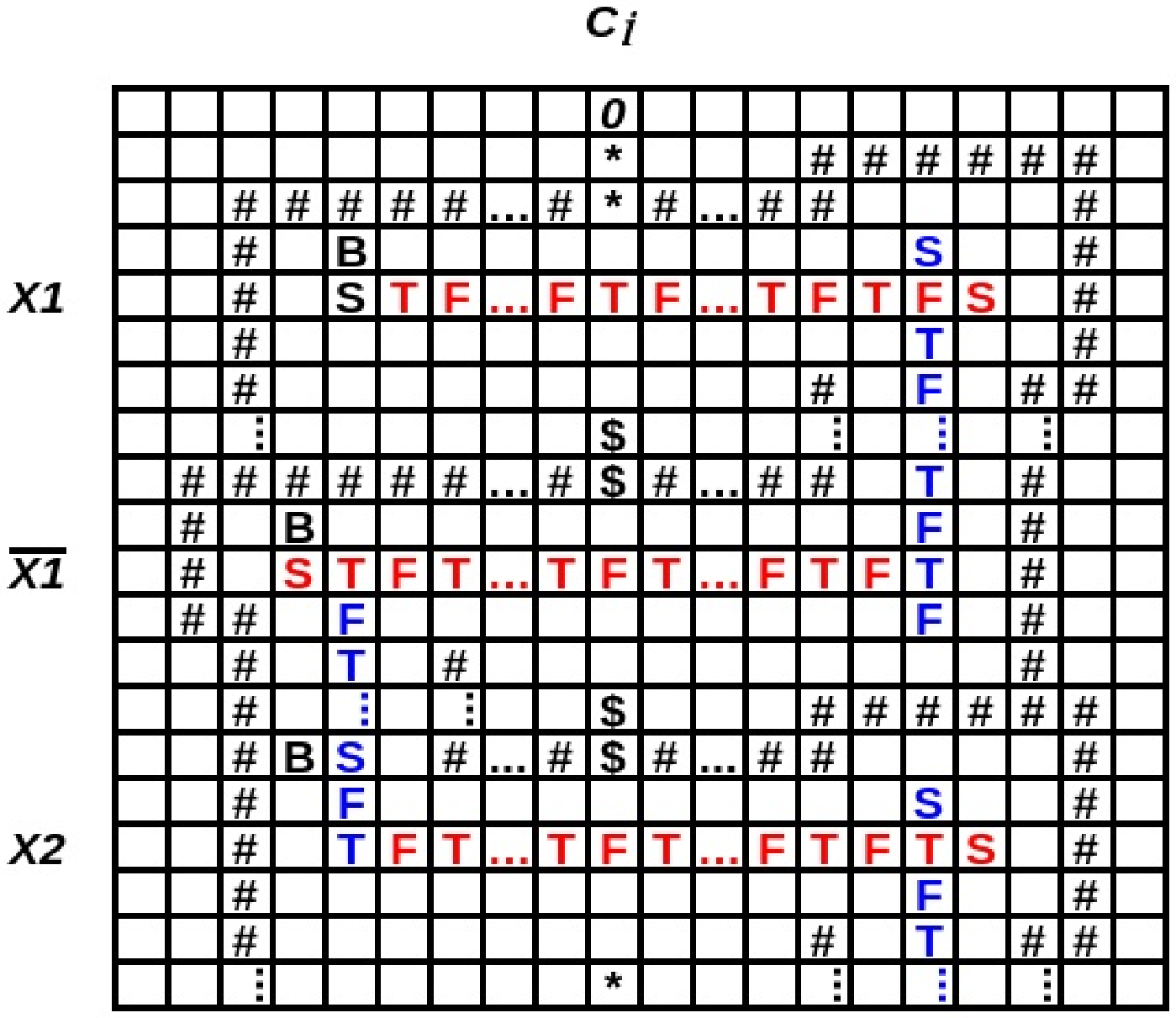}}
\subfloat[The satisfaction phase]
{\label{clauses}\includegraphics[width=0.6\textwidth]{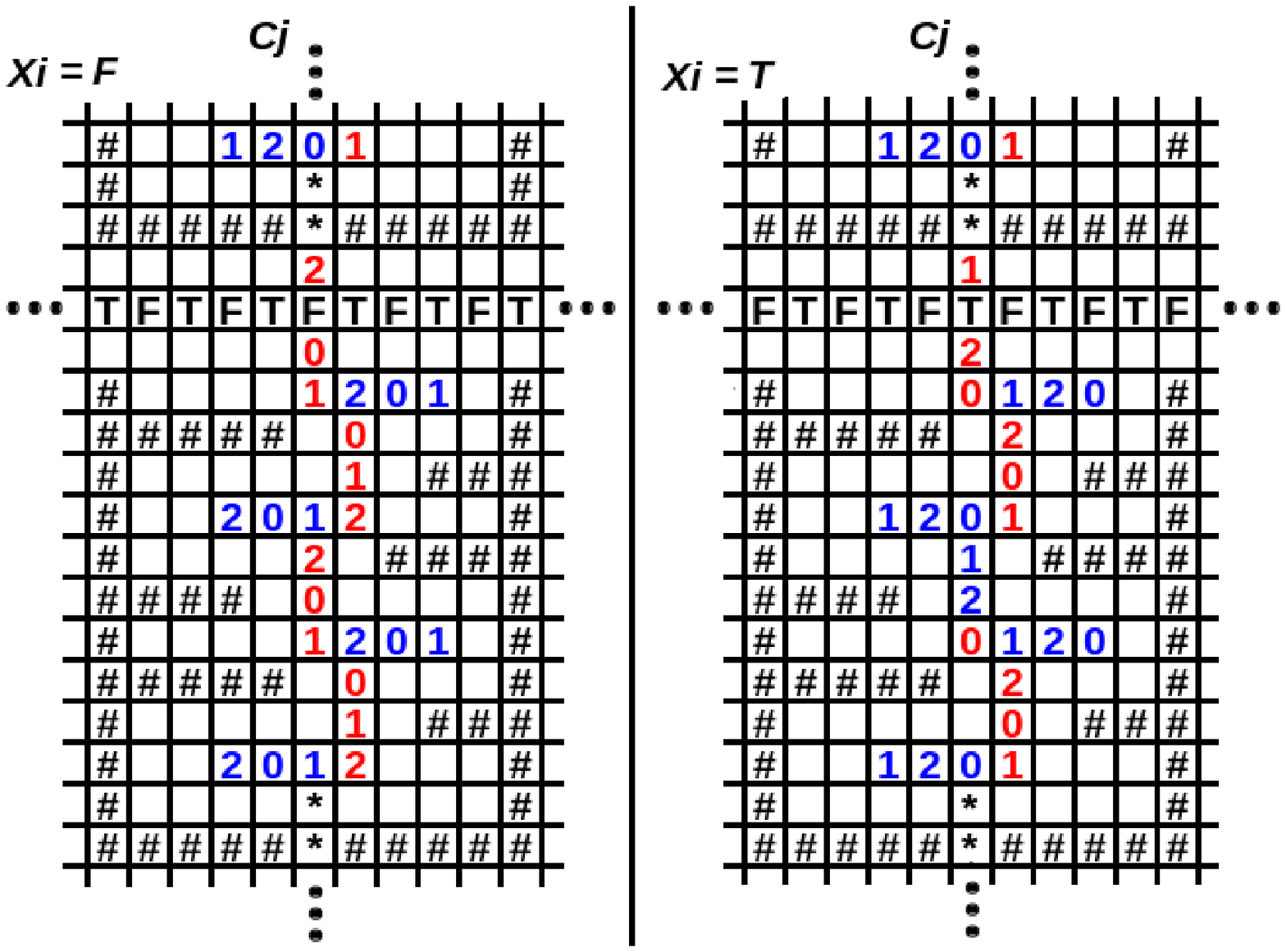}}

\caption{More detailed construction sketches.}
%\end{figure}

\end{figure}

In the first phase of the game (the assignment phase, see figure
\ref{variables} or \ref{var_form} in the appendix), players will repeatedly
draw the following letters: $\frac{k-1}{2}$ pairs (\s{T}s,\s{F}s) and a single
\s{S}. The only words that they can form with these symbols are the assignment
words from $\Delta$ (given in the first two lines in the dictionary of table
\ref{dictionary}). These words have length $k+1$, so in order to play them, one
of the symbols should already be on the board in the same line as the word
placement and the players have to empty their racks completely.

The major concern here is the assignment. We say that a word assigns the value
True (resp. False) to a variable if the intersection of the positive literal's
line with the clauses columns contain the symbol \s{T} (resp. \s{F}). 

Player 1 plays first and has to choose among two possible proper plays, one
that assigns the value True to $x_1$ and one the value False. Observe that
player 1 is always forced to play horizontally whereas player 2 only plays
vertically. To avoid having only player 1 choose the assignment, we use the
symbols \s{A}, \s{B} and \s{S}, as described in the appendix. 

%The last important detail is the meaning of the symbols \s{S} and the roles of
%the symbols \s{A} and \s{B} which are attached on the left wall. The symbol
%\s{S} in general indicates that the player has the choice to continue either
%with (TF) or with (FT), choosing thus the assignment. Now, symbol \s{B}
%enforces only an \s{S} attaced to it, which gives the next player the ability
%to reset the assignment by using one of the reset words (last two items on the
%first two lines of the dictionary). On the other hand, symbol \s{A} enforces a
%\s{T} or \s{F} symbol attached next to it, which leaves symbol \s{S} at the end
%of the played word and thus giving the current player the chance to choose
%among two possible proper plays, determining the truth value of the new
%variable. 

Once the assignment is fixed, players' unique choices are predetermined by the
current position of the board and the dictionary. The amount of points that the
two players gain after this phase is identical and equal to $2n(2k+5)$ (there
are $2n$ zigzags and each player constructs two $(k+1)$-letter long words and
one 3-letter long word in each). 

%Figure \ref{variables} shows the assignment $x_1 = T, x_2 = F, x_3 = T$.

\subsection{Satisfaction Phase}\label{satisfaction}

For this section, refer to figure \ref{clauses} (or to the more detailed preview
 \ref{clause_form} in the appendix).

%\begin{figure}
%\centering
%\includegraphics[scale=0.25]{clauses}
%\caption{The satisfaction phase}
%\label{clauses}
%\end{figure}

After the assignment phase, the bag begins with a long string of the symbols
\s{0}, \s{1}, \s{2}. %$\bag = (\s{0}\s{1}\s{2})^a$.  
Satisfaction is realized by forming satisfaction words (the last four lines in
the dictionary). A clause is considered satisfied when the corresponding
vertical segment is fully filled with words.

The most crucial step of the satisfaction phase is the placement of the words
that intersect with literals. The numbers \textbf{0, 1, 2} indicate the number
of false literals the clause currently has. The combination of \{\s{*},
\s{\$}\}, \{\s{T}, \s{F}\} and \{\s{0}, \s{1}, \s{2}\} gives a unique vertical
proper word to play in the intersection of a literal (horizontal) segment with
the clause (vertical) segment. The ending symbol of the played word is the
number of false literals we have seen in the clause so far. The combination
\{\s{num}, \s{*}, \s{F}\} (where \s{num} = \s{0}, \s{1},or \s{2}) is important,
because it forms the word \s{num}\s{*}\s{*}\ldots\s{F}\ldots\s{num}+1 which is
the only one that increases \s{num} (the clause contains a false literal).

The words which contain only the symbols \s{0}, \s{1}, \s{2} reserve the order
of their appearence and by doing so enforce the appropriate number to begin the
next intersection word. 

Starting with literal $x_1$, the two players fill in words interchangeably,
beginning with player 1 who plays vertically. Observe that the only way that a
player won't be able to place a word is to be faced with the combination
\{\s{2},\s{*},\s{*}, \s{F}\} in an intersection (third false literal in the
clause).

Notice that player 2 doesn't really have an incentive to play vertically
because the number of points acquired if she plays vertically is equal to the
number of points if she plays horizontally and equal to $\frac{4l}{2}+3 =
2l+3$, where $l=\frac{s}{2nm}$ is the number of turns played inside a literal
segment (the additive term in the score comes in the vertical play case from
the 7-letter long word played during the first turn and in the horizontal play
case from the additional 3-letter long word which is formed during the last
turn).  Thus we can assume wlog that player 1 plays vertically and player 2
horizontally, and, despite that during the game there will be several possible
proper plays, the final score after the satisfaction phase is independent of
players' choices.

%Figure \ref{clauses} shows the satisfaction part for variable $x_1$. 

%\subsection{Winner} \label{sec:winner}

We argue now that if there is a satisfying assignment for the first order
formula then player 1 wins, else player 2 wins.

The key point in this proof is that player 2 ``matches'' player 1's moves
throughout the duration of the whole game. Since player 2 starts with a 1-point
lead she will continue to have the lead after the end of the satisfaction
phase. 

If there is a satisfying assignment, then by the end of the game player 1 gets
the last symbol in the bag which is an \s{A} and forms an additional 3-letter
long word, which makes him the winner of the game with $\score{1}=\score{2}+2$.

On the other hand, if there is no satisfying assignment the two players will
have at least one set of \s{0},\s{1},\s{2} on their hands and probably some
copies of the useless symbol \s{@} which doesn't form any words, so player 1 is
not going to get the symbol \s{A} from the bag. Player 2 is the last player to
place a word on the board.  This makes him the winner of the game with
$\score{2}= \score{1}+1$.

Let us also point out that the fact that we assumed players cannot pass does
not affect our arguments so far. Indeed, observe that at any point when it's a
player's turn to play, that player is behind in the score. If she chooses to
pass, the other player may also pass. Repeating this a second time ends the
game, according to standard Scrabble rules. Thus, if the current player has a
winning strategy it must be one where she never chooses to pass.

\subsection{Constant rack and word size} \label{sec:constant}

In order for the proof to work for constant size words and rack, we need to
break the long assignment words into constant size ones and zig-zag through the
clauses (see figure \ref{fixed}).  Once we reduce the size of the words to a
constant an unbounded size rack is unnecessary. In fact the rack has to be
smaller than the maximum word size by one symbol. 

\begin{figure}
\centering
\includegraphics[scale=0.20]{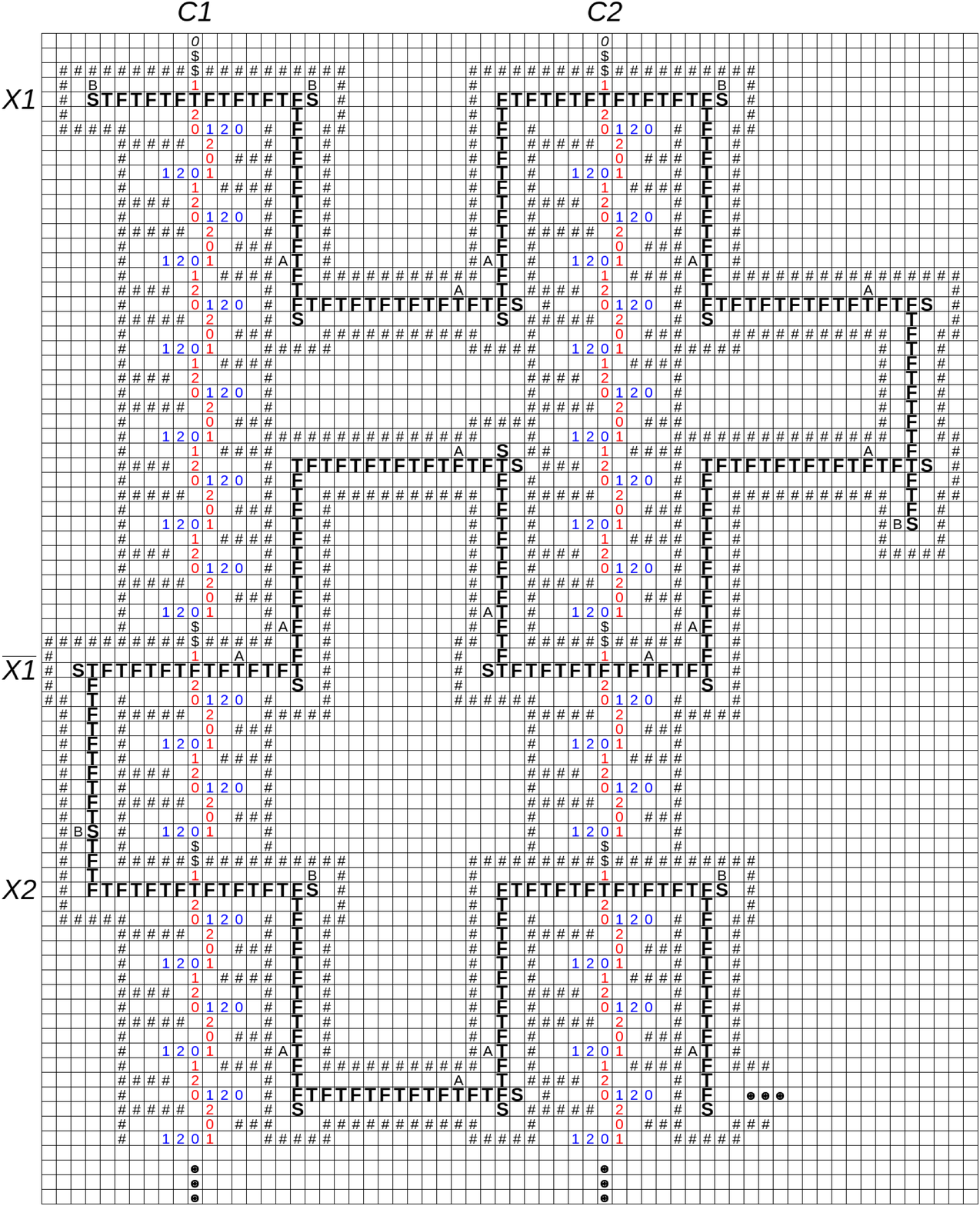}
\caption{Modifications for fixed size words and rack}
\label{fixed}
\end{figure}

Observe that the length of the assignment word should be equal to the height of
the clause segments between a negative literal and its next positive. This
distance is $4 \bmod 6$.  Also, the word has to be longer than the width of the
clause segments (which is 11). Setting the word size equal to 16 ($k=15$),
satisfies both requirements. Careful counting arguments fix the zig-zaging
between a positive and a negative literal (see figure \ref{fixed}). 

We change the board construction to adopt the modifications:

\begin{itemize}

\item We build walls all around the board to force the aforementiond zig-zaging
pattern. The walls too have to consist of constant size parts (the wall is part
of the dictionary).  

\item Last, we need to place one \s{A} or \s{B} symbol in every horizontal or
vertical section, so that we force the players to put their $\s{S}$ in the
beginning or the end of their played word (forcing thus the assignment
throughout variable segments) and also to make sure that the players will gain
an equal amount of points ($=k+3$) on every turn.

\end{itemize}

The rest of the proof follows the ideas of the proof for arbitrary size rack and 
words.

\end{proof}

\section{Conclusions}\label{conclusions} 
We have established the PSPACE-hardness of (deterministic) Scrabble in two
different ways. The main ingredients for our two proofs are the possibility of
placing words in many places in the first, and the possibility of forming
several different words in the second. We have also established that hardness 
remains even when all relevant parameters are small constants.

Several interesting further questions can be posed in the same vein. Are the
constants we have used optimal? What is the minimum-size alphabet or dictionary
for which the problem is still PSPACE-hard? In particular, does the problem
become tractable when the alphabet contains just one letter, or is the
complexity of placing the tiles on the board enough to make the problem
hard?

Another interesting question was posed by Demaine and Hearn
\cite{demaine2005playing}: is there a polynomial-time algorithm to determine
the move that would maximize the score achieved in this round? Of course, in
the case of a bounded-size rack the problem is immediately in P, but deciding
how to place $n$ letters on the board optimally could be a much harder problem.

\bibliography{scrabble}

\appendix
\newpage

\section{Omitted proofs and figures}

\subsection{Proof of Facts \ref{eatsall} and \ref{fact2}}
% \ref{fact3}}

\begin{proof}

%Let us notice first that at every round for each word on the board, its length
%is either 2 or it is greater or equal than $2r+2$. At the start of the game
%this is obviously true, and at every round we cannot form a word shorter than
%$2r+2$ (because all the words in the dictionary are of length $2r + 2$).
%Therefore the only way to form a legal word is to extend an existing 2-letter
%word on the board, which uses up all $2r$ letters from the rack.

The dummy words appearing on the board in the beginning of the game are:

\begin{itemize} 

\item \s{\#}$^{4r+3}$;
\item \s{@}\s{$x_i$};
\item \s{$x_i$}\s{@}; 
\item \s{\#}$^r$\s{$x_i$}\s{\#}$^{r-1}$;
\item \s{\#}$\prod_{u=1}^p$(\s{\#}\s{$\test{i}{j_u}{}$})\s{\#}$\prod_{v=1}^p$(\s{$\test{i}{k_v}{\neg}$}\s{\#})\s{\#} and
\item \s{@}\s{$\test{i}{j}{(\neg)}$},

\end{itemize}

\noindent for every variable $x_i$, and for every appearance of variable $i$ in the clauses.

Let us notice that as long as the contents of the player's rack consist only of
the symbols $\{\mathbf{x_i},\$ \}$ or $\{\test{i}{j}{(\neg)},\ast\}$, the only
possible word which can be formed has length $2r+2$, since the \s{@} characters
are unusable.  Thus, in the value-assigning phase players must use all their
racks and two letters from the board to form words.  This establishes Fact
\ref{eatsall}.

From the previous fact we gather that during each round in the
value-assigning	phase, the contents of the player's rack are
$\,\mathbf{x_i\$}^{2r-1}$ for some $i$. A simple case by case analysis shows
that the player can form a word from these letters only in one of the two ways
shown in Figure \ref{fig:assignments}. This establishes Fact \ref{fact2}.

%Basing on the previous two facts we know that during each round
%in the test phase, the contents of the player's rack are
%$\{\test{a}{j}{(\neg)}, \test{b}{j}{(\neg)}, \test{c}{j}{(\neg)}$,
%$\ast^{2r-3}$\} for $a$, $b$ and $c$ being the indices of the literals
%appearing in clause $j$.  One can easily see that the player can form a legal
%word from these letters only by extending one of the 3 words
%\s{@}$\test{v}{j}{(\neg)}$, $v \in \{a, b, c\}$, by arranging her symbols in a
%permutation $(a', b', c')$ where $v = a'$.
%
%The player can choose any of such permutations, which means she can choose the
%literal, in whose gadget she will play the word. A simple analysis shows that
%the player can play this word in that position iff the valuation of the
%variable agrees with the chosen literal (i. e. if the chosen literal reads
%$\neg x_j$, then $x_j$ must have been set to false etc.).

\end{proof}				

\subsection{Proof of Theorem \ref{thm:pspace1}}

\begin{proof}

Given is a first order formula $\exists x_1 \forall x_2 \ldots \phi$, with $n$
variables and $m$ clauses. We can assume that $n$ is even. If it is not we just
add in $\phi$ a new dummy clause in which a new variable $x_{n+1}$ will appear
both positive and negated. 

We first create a propositional formula $\phi'$ by duplicating all clauses from
$\phi$.  Observe that the new instance of \QBF\ $\exists x_1 \forall x_2 \ldots
\phi'$ is equivalent to the original.

It is easy to reduce the new instance of \QBF\ to a game of Scrabble
$\mathcal{S}$. The alphabet $\Sigma$, the dictionary $\Delta$, the rack size
$k$, the board construction $\board$ are defined in the same way as in proof of
lemma \ref{lemma:np}. The bag sequence $\bag$ is again defined almost
identically apart from the addition of the symbol \s{@} in the very end of it.
The scores are $\score{1} = \score{2} - 1$ (i.e.  player 2 has a lead of 1
point) and it is first player's turn.

The two players are going to play the normal game of Scrabble (starting by
player 1) in a the board obtained if we apply the previous construction to the
duplicated formula.  It is easy to observe that, while the number of variable
gadgets is the same, their sizes are doubled since each literal appears in
twice as many clauses as in $\phi$. 

In the assignment phase, the two players will assign truth values to the
variables $x_1$, $x_2$, \ldots, $x_n$ interchangeably. Since $n$ is even,
player 2 is the last player to put an assignment word on the board, leaving
player 1 to begin phase 2.

For the satisfaction part, observe that, for every clause $c_u$ there is an
indentical clause $c_u'$. If there is a literal $l_i$ that satisfies $c_u$ then
$l_i$ also satisfies $c_u'$. That means that player 2 cannot be left without an
available word to play since she can always match player 1's placement. 

If the formula is satisfiable then the bag will eventually empty (as it was
shown in section \ref{placement}) and the last player to place a word will be player
1, using the symbol \s{@} to create a two-letter word. In this case player 1
wins with $\score{1} = \score{2} + 1$.

On the other hand, if the formula is not satisfiable, then the last player to
place a word will be player two, leaving the score $\score{1} = \score{2} - 1$
and making player 2 the winner of the game.  

\end{proof} 

\subsection{Omitted details from section \ref{assignment}}

In order to enforce the two players to assign values to the variables
interchangeably we need to use the symbols \s{A} and \s{B} (attached to the
left wall) (see figure \ref{board}). The place where player 2 is going put 
the \s{S} symbol that holds
in her rack when she plays her vertical word on the left side of the board
specifies which player's turn is to choose the truth value of the next
variable. The symbol \s{S} indicates that the player has the choice to continue
either with (\s{T}\s{F}) or with (\s{F}\s{T}), choosing thus the assignment.
Now, symbol \s{B} enforces only an \s{S} attaced to it (forming the valid word
\s{\#}\s{B}\s{S}), which gives player 2 the ability to reset the assignment by
using one of the reset words (last item on the first two lines of the
dictionary). On the other hand, symbol \s{A} enforces a \s{T} or \s{F} symbol
attached next to it (forming one of the valid words \s{\#}\s{A}\s{T},
\s{\#}\s{A}\s{F}), which leaves symbol \s{S} at the end of the played word and
thus giving player 1 the chance to choose among two possible proper plays,
determining the truth value of the new variable (see figure \ref{var_form}).

\subsection{Omitted figures}

\begin{figure}

\centering
\includegraphics[scale=0.6]{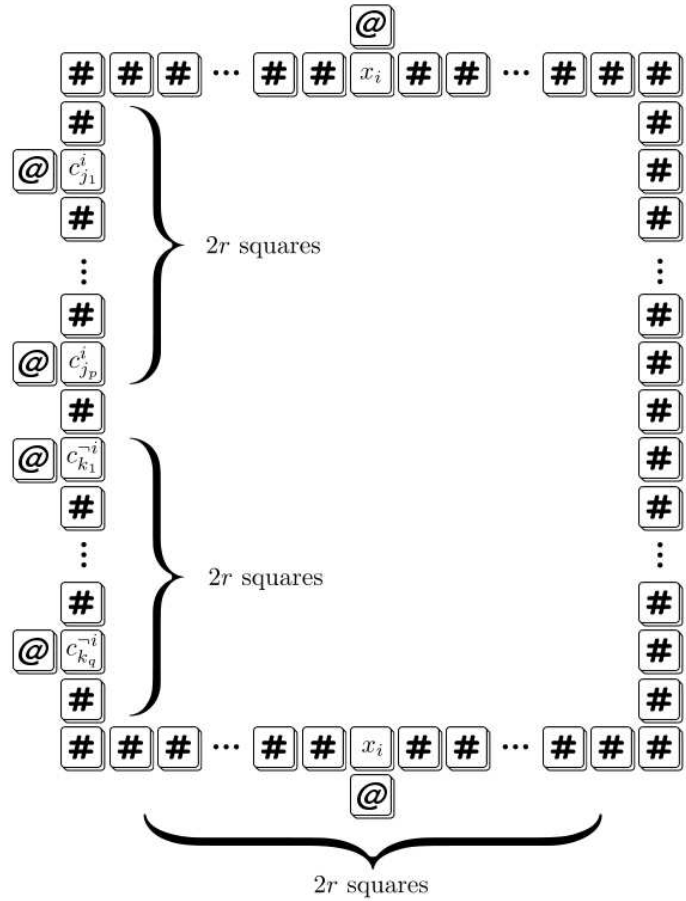}
\caption{The gadget for variable $x_i$ (section \ref{placement}).}
\label{var_place}

\end{figure} 

\begin{figure}

\centering
\includegraphics[scale=0.5]{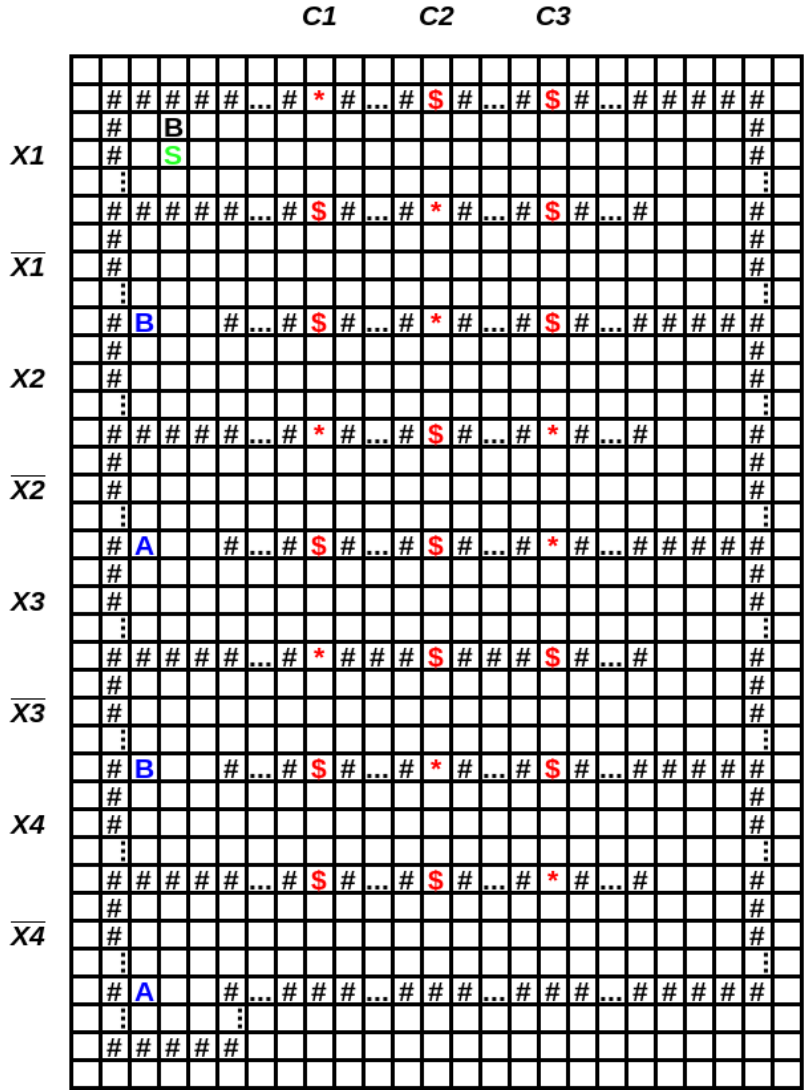}
\caption{An abstract construction of the board for $\phi = (x_1 \vee \neg{x_2} \vee
\neg{x_3}) \wedge (\neg{x_1} \vee x_2 \vee x_4) \wedge (\neg{x_2} \vee x_3
\vee \neg{x_4})$ (section \ref{formation}).}
\label{board}

\end{figure} 

\begin{figure}

\centering
\includegraphics[scale=0.5]{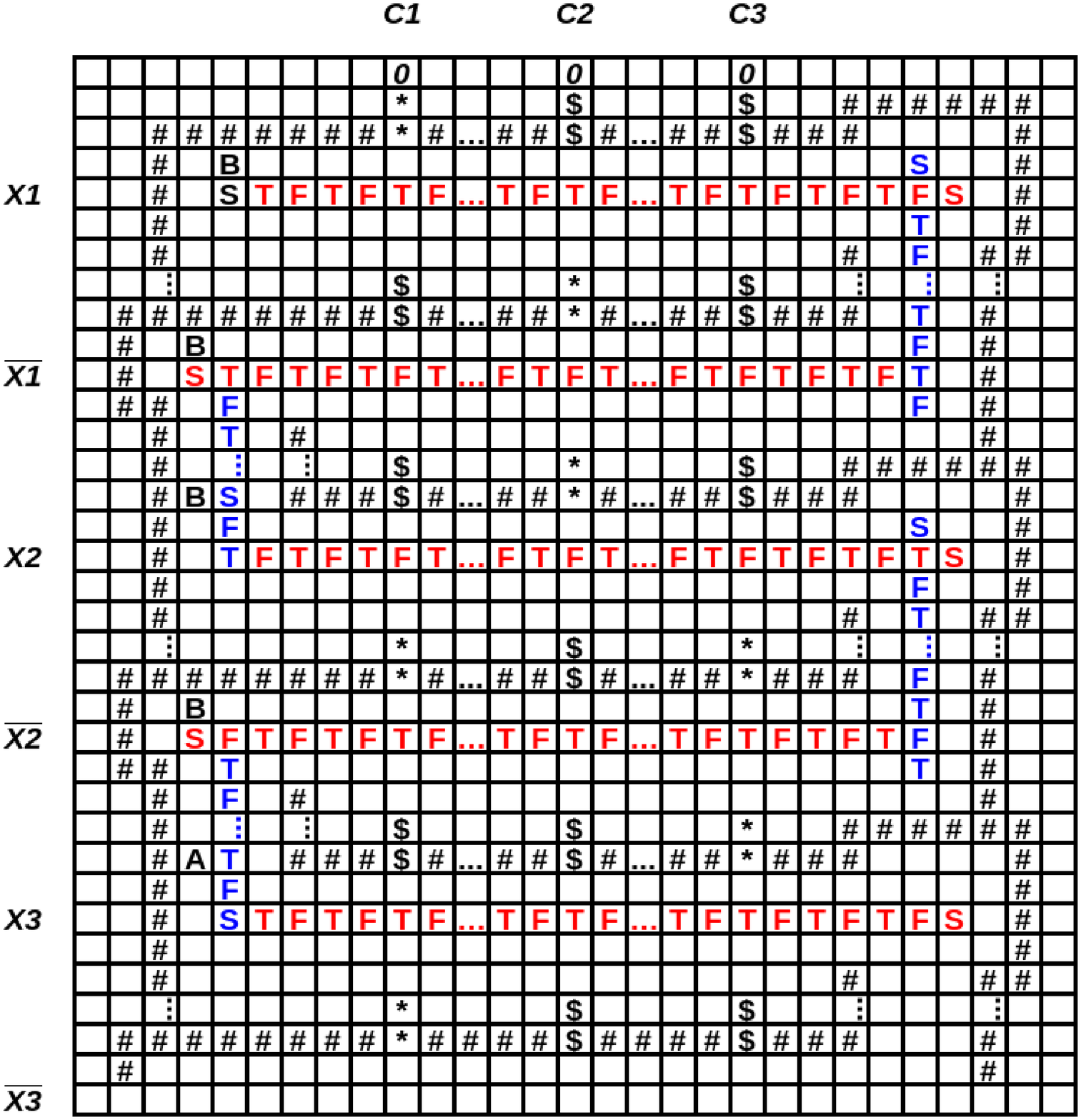}
\caption{A more detailed view of the board for the assignment phase of section
\ref{formation}. In this example $x_1 = T, x_2 = F, x_3 = T$.}
\label{var_form}

\end{figure} 

\begin{figure}

\centering
\includegraphics[scale=0.4]{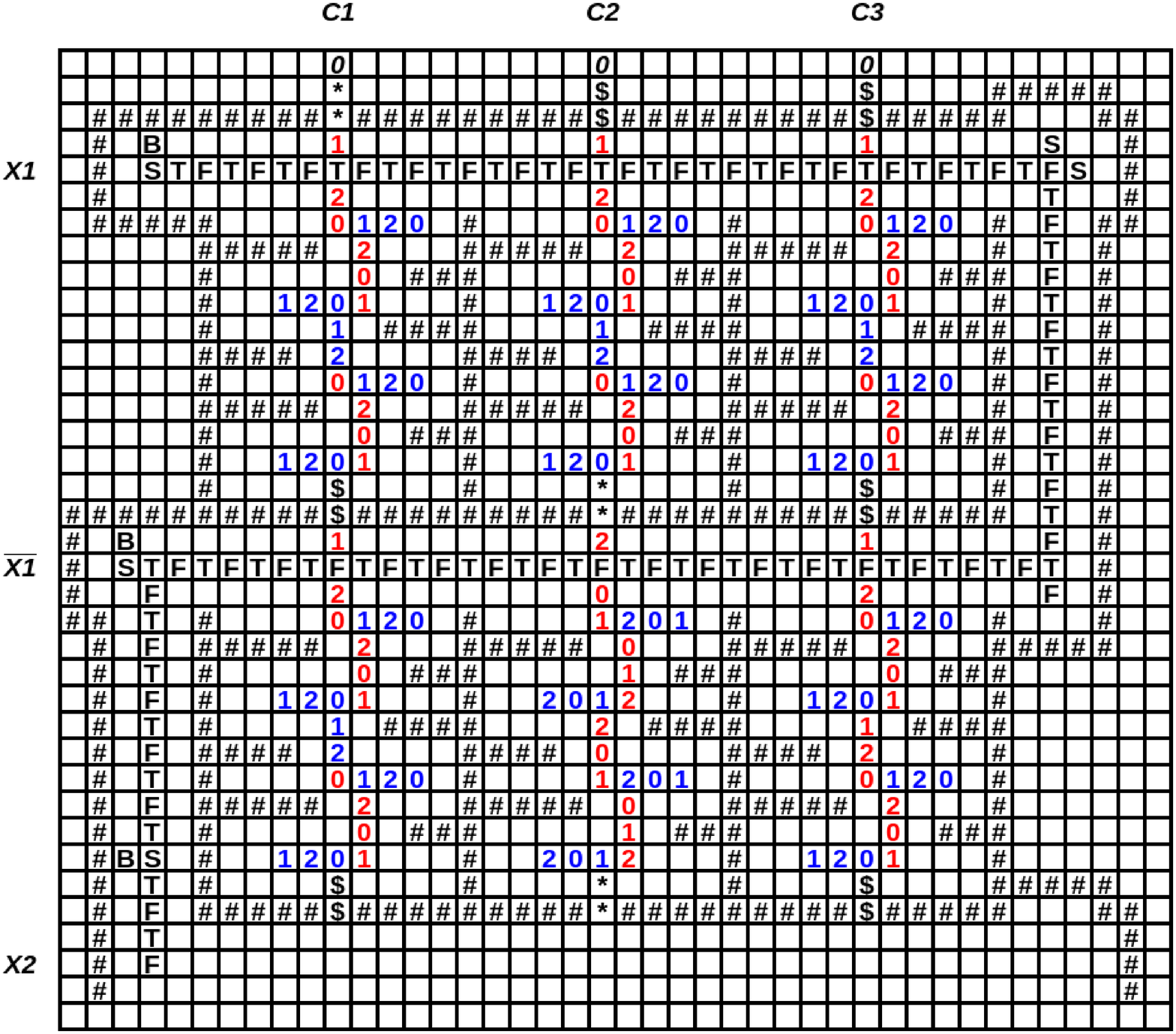}
\caption{A more detailed view of the board for the satisfaction phase of section 
\ref{formation}.}
\label{clause_form}

\end{figure}

\end{document}